\documentclass[runningheads]{llncs}

\usepackage{amssymb}
\usepackage{amsmath}
\usepackage{t1enc}
\usepackage[utf8]{inputenc}
\usepackage{ae,aecompl}

\usepackage[shortcuts]{extdash}
\usepackage{tikz}
\usetikzlibrary{positioning, shapes, decorations.markings, arrows,backgrounds,calc}
\usepackage{color, soul}
\usepackage{marginnote}
\usepackage{array,booktabs,calc}
\usepackage{stmaryrd}
\usepackage{todonotes}
\usepackage[caption=false]{subfig}
\usepackage{wrapfig}
\usepackage[inline]{enumitem}
\usepackage{hyperref}
\usepackage[SEQ]{prftree}

\usepackage{euler}
\usepackage{sequents}
\usepackage{logics}

\definecolor{greylight}{rgb}{0.7,0.7,0.7}
\definecolor{greyblue}{rgb}{0.0,0.0,0.5}
\definecolor{greyred}{rgb}{0.5,0.0,0.0}
\definecolor{lightgreyblue}{rgb}{0.5,0.5,0.9}
\definecolor{lightgreyred}{rgb}{0.6,0.2,0.2}
\definecolor{lightgreyorange}{rgb}{0.7,0.5,0.2}

\newlength\savedintextsep 

\title{Intuitionistic Euler-Venn Diagrams (extended)\thanks{This work was
supported by EPSRC Research Programme EP/N007565/1
\emph{Science of Sensor Systems Software}.}}

\author{Sven Linker }
\institute{University of Liverpool, UK\\
\email{s.linker@liverpool.ac.uk}
}

\begin{document}

\maketitle

\begin{abstract}
  We present an intuitionistic interpretation
  of Euler-Venn diagrams with respect
  to Heyting algebras. In contrast to classical Euler-Venn
  diagrams, we treat shaded and missing zones differently,
  to have diagrammatic representations of conjunction, disjunction
  and intuitionistic implication. We
  present a cut-free sequent calculus for this language, and prove it
  to be sound and complete.  Furthermore, we show that the rules of
  cut, weakening and contraction are admissible.
\end{abstract}

\section{Introduction}
\label{sec:intro}


Among diagrammatic systems to reason about logic, Euler-Venn circles have
a long tradition. They are
known to be a well-suited visualisation of classical
propositional logic. In previous work \cite{Linker2018}, we have presented a
proof system in the style of sequent calculus \cite{Gentzen1935} to
reason with Euler-Venn diagrams. There, we speculated that, similar to sentential
languages, restricting the rules and sequents in the system would
allow for intuitionistic reasoning with Euler-Venn diagrams. However,
further investigation showed that such a simple change is not
sufficient, due to the typical use of the syntax elements of Euler-Venn diagrams.

Consider for example the diagrams in Fig.~\ref{fig:example-equivalence}. In the classical
interpretation, these diagrams are equivalent: the shaded zone in Fig.~\ref{fig:example-vd}
denotes that the situation that \(a\) is true and \(b\) is false is prohibited,
which is exactly what the omission of the zone included in the contour \(a\), but not in \(b\) in Fig.~\ref{fig:example-ed}
signifies as well. 
\setlength\intextsep{0pt}
\begin{wrapfigure}[8]{l}{.35\textwidth}
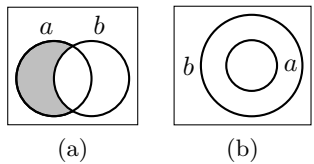

  \centering
\subfloat[]{ 
  \subsetVD{a}{b}{1cm}
  \label{fig:example-vd}
}
\subfloat[]{ 
  \subsetED{b}{a}{1.35cm}
  \label{fig:example-ed}
}
  \caption{Euler-Venn diagrams}
  \label{fig:example-equivalence}
\end{wrapfigure}
\setlength{\intextsep}{\savedintextsep}
That is, shading a zone and omitting it is equivalent in classical
Euler-Venn diagrams. Additionally, we can interpret these two
diagrams in two ways: Fig.~\ref{fig:example-vd} may intuitively
be read as \(\lnot ( a \land \lnot b)\): we do not
allow for the valuations satisfying \(a\), but not \(b\). Fig.~\ref{fig:example-ed}, however, is more
naturally read as \(a \implies b\): whenever a valuation satisfies
\(a\), it also satisfies \(b\). While in a classical interpretation,
these two statements are indeed equivalent, they are generally
not equivalent in an intuitionistic interpretation. Hence, we want
to treat missing zones and shaded zones differently. Since
typically, proof systems for Euler diagrams allow to change
missing zones into shaded zones \cite{Linker2018,Howse2005,Stapleton2007}, this
implies a stronger deviation from our sequent calculus rules than anticipated.

Furthermore, we want to emphasise a constructive approach to reasoning. In particular,
instead of emphasising a \emph{negative} property by prohibiting 
interpretations of the diagrams, we will treat shading as a \emph{positive}
denotation. While this would not make much of a difference in a classical system,
negation in intuitionistic systems is much weaker, and hence not suited
as a basic element for the semantics of a language.

In this paper, we present an intuitionistic interpretation of Euler-Venn diagrams that
takes the preceeding considerations into account. To that end, we will
distinguish between \emph{pure Venn}, \emph{pure Euler} and \emph{Euler-Venn}
diagrams, and present intuitionistic interpretations of these types of
diagrams based on Heyting algebras. Subsequently, we present a proof system
in the style of sequent calculus, which we prove to be sound and complete.
Furthermore, we show that the structural rules of weakening, contraction
and cut are admissible. 

\paragraph{Related Work.} 
Many reasoning systems for visualisations of classical
logic have been defined over time, for example the initial work
of Venn~\cite{Venn1881} and Peirce~\cite{Peirce1931} and subsequently
the work of Shin~\cite{Shin1995} and Hammer~\cite{Hammer1996}, as well
Spider diagrams by Howse et al.~\cite{Howse2005}. Most of these
systems are not directly comparable to sentential reasoning systems,
due to very different structure of the rules, with the notable exception
of the work by Mineshima et al.~\cite{Mineshima2010} and Takemura~\cite{Takemura18}.

However, the situation is different for non-classical logics. There
are several visual reasoning systems for non-classical variants of
Existential Graphs. For example, Bellucci et al. defined
\emph{assertive graphs}~\cite{Bellucci2018}, including
a system based on rules for iteration and deletion of graphs, among others.
This logical language reflects intuitionistic logic, but the rules manipulate only
single graphs, while sequent calculus systems manipulate sequents of diagrams.
Ma and Pietarinen presented a graphical system for intuitionistic logic \cite{Ma2019}
and proved its equivalence with Gentzen's
single
succedent sequent calculus for propositional intuitionistic logic. To that end,
they translate the graphs into sentential formulas. They also extended their approach
to existential graphs with quasi-Boolean algebras as their semantics~\cite{Ma2018}.
Legris pointed out that structural rules of sequent calculi can
be seen as special instances of rules in the proof systems for existential
graphs, to analyse substructural logics~\cite{Legris2018}. de Freitas and
Viana presented a calculus to reason about intuitionistic equations \cite{deFreitas2012}.
However, we are not aware of any intuitionistic reasoning systems using Euler-Venn-like
visualisations.

\paragraph{Structure of the paper.} Following this introduction,
we briefly recall the foundations of intuitionistic logic
and its semantics in terms of Heyting algebras in Sect.~\ref{sec:intuition}.
In Sect.~\ref{sec:diagrams}, we define the system of Euler-Venn diagrams,
followed by the graphical sequent calculus system,
as well as soundness and completeness proofs, in Sect.~\ref{sec:sequent}.
Section~\ref{sec:admissible} contains proofs for the admissibility
of the structural rules.
Finally, we discuss our system and conclude the paper in Sect.~\ref{sec:conc}.


\section{Intuitionistic Logic}
\label{sec:intuition}
In this section, we give a very brief overview of the
aspects of propositional intuitionistic logic we will use. We start by
presenting the underlying semantical model we use,
Heyting algebras.

\begin{definition}[Heyting Algebra]
  \label{def:heyting}
  A \emph{Heyting algebra} \(\model = (H, \hor, \hand, \himplies, 0, 1)\) is a bounded, distributive
  lattice, where \(\hor\) is the join,
  \(\hand\) the meet, \(0\) the bottom and \(1\) the top element of the lattice. Observe that
  such a bounded lattice  possesses a natural partial order \(\leq\) on its elements. The
  binary operation \(\himplies\), \emph{the implication}, is defined by \(c \hand a \leq b\)
  if, and only if, \(c \leq a \implies b\). That is, \(a \implies b\) is the join
  of all elements \(c\) such that \(c \hand a \leq b\).
  We will use the abbreviation
  \(\hnot a\) for \( a \himplies 0\).
  Furthermore, we set \(\hbigand_{i \in \emptyset} a_i = 1\) and \(\hbigor_{i \in \emptyset} a_i = 0\) for any
  \(a_i\).
\end{definition}

\newcounter{heyting}
\newcommand{\heytingprop}{\refstepcounter{heyting}\textnormal{(\arabic{heyting})}}

We collect a few basic properties of Heyting algebras that we need in
the following. Proofs can be found, e.g.,  in the work of Rasiowa and Sikorski \cite{Rasiowa1963}.

\begin{lemma}[Properties of Heyting Algebras]
  \label{lem:prop}
  Let \(\model\) be a Heyting algebra. Then for all elements \(a\), \(b\) and \(c\),
  we have
  \vspace{.5em}
  
  \begin{tabular}{p{.25\textwidth}p{.25\textwidth}p{.4\textwidth}}
    \(a \hand (a \himplies b)  \leq b\) \heytingprop\label{eq:mp}
    &
      \(       (a \himplies b) \hand b  = b\) \heytingprop \label{eq:imp_b}
    &
      \(    a \himplies ( b \himplies c)  = (a \hand b ) \himplies c\) \heytingprop \label{eq:impl_and}
  \end{tabular}
\end{lemma}

The syntax of propositional intuitionistic logic is similar to classical Boolean logic,
with the difference that the operators are not interdefinable. Hence, the signs
for conjunction,
disjunction, and implication are all necessary as distinct symbols, and
cannot be treated as abbreviations. We will assume a fixed, countable
set of propositional variables \(\vars\).
\begin{definition}[Syntax]
  \label{def:logic_syn}
  An intuitionistic formula is given by the following EBNF
  \begin{align*}
    \formula \colon = \false \mid p \mid \formula \land \formula \mid \formula \lor \formula \mid \formula \implies \formula \enspace, \text{where } p \in \vars \enspace.
  \end{align*}
\end{definition}

We will treat negation as the abbreviation \(\lnot \formula \equiv \formula \implies \false\).
Furthermore, we let \(\true \equiv \false \implies \false\). The semantics of a formula is based on valuations,  associating each variable
with an element of a given Heyting algebra.
\begin{definition}[Semantics]
  \label{def:logic_sem}
  Let \(\model \) be a Heyting algebra and \(\valSing \colon \vars \to H\) a \emph{valuation}, mapping  variables to elements of \(\model\).
We lift valuations to formulas.
  \begin{align*}
    \val{\false} & = 0\\
    \val{\formula \land \formulaTwo} & = \val{\formula} \hand \val{\formulaTwo}\\
    \val{\formula \lor \formulaTwo} & = \val{\formula} \hor \val{\formulaTwo}\\
    \val{\formula \implies \formulaTwo} & = \val{\formula} \himplies \val{\formulaTwo}
  \end{align*}
  A formula \(\formula\) \emph{holds} in \(\model\), if \(\val{\formula} = 1\). If \(\varphi\)
  holds for every valuation of \(\model\), we write \(\model \models \formula\). If \(\model \models
  \formula\) for every Heyting algebra \(\model\), we say that \(\formula\) is \emph{valid}.
\end{definition}

\section{Euler-Venn Diagrams}
\label{sec:diagrams}
In this section, we present the syntax and semantics of
Euler-Venn diagrams with an intuitionistic interpretation.
Generally, a diagram can be \emph{unitary} or \emph{compound}.
A unitary diagram consists of a set of \emph{contours}
dividing the space enclosed by a 
bounding rectangle into different \emph{zones}. Zones
may also be shaded. Depending on how the contours
may be arranged, and whether zones may be shaded,
we distinguish between \emph{Venn} diagrams,
\emph{Euler} diagrams, and \emph{Euler-Venn} diagrams.
Compound diagrams are constructed recursively. Since
the structure of compound diagrams is the same, regardless
of the type of unitary diagrams, we present their syntax
first.
\begin{definition}[Compound Diagrams]
  \label{def:compound-syntax}
  A \emph{compound diagram} is created according
  to the following syntax,
 \begin{align*}
   D &::= d \mid   D \land D \mid D \lor D \mid 
 D \implies D \, ,
 \end{align*}
where \(d\) is a
 unitary diagram.
\end{definition}

\begin{definition}[Compound Diagram Semantics]
  \label{def:compound-semantics}
  The  semantics of compound diagrams for a Heyting algebra \(\model\) and a valuation \(\valSing\)
  is given as follows.
 \begin{align*}
   \val{D_1 \land D_2} & = \val{D_1} \hand \val{D_2} \\
   \val{D_1 \implies D_2} & = \val{D_1} \himplies \val{D_2}  \\
   \val{D_1 \lor D_2} & = \val{D_1} \hor \val{D_2}
 \end{align*}
 where 
 \(D_1\), \(D_2\) are compound diagrams. If \(\val{D} = 1\), for all intuitionistic models \(\model\)
 and valuations \(\valSing\) then we call \(D\) \emph{valid}.
\end{definition}
Observe that we did not give the semantics for unitary diagrams in the previous
definition. While we will fill this gap in the next sections, we first
present notations that are used for all types of diagrams alike.
Formally, a \emph{zone}
for a finite set of contours \(\contours \subset \vars\) is a tuple
\((\zin, \zout)\), where \(\zin\) and \(\zout\) are disjoint 
subsets of \(\contours\) such that \(\zin \cup \zout = \contours\).
We will also write \(\zinF{z}\) and \(\zoutF{z}\) to refer to
the corresponding sets of contours in \(z\).
The set of all possible zones for a given set of contours is denoted
by \(\vennzones{\contours}\).

\paragraph{Venn Diagrams}
\label{sec:venn-def}
A \emph{Venn diagram} is a
diagram where all possible zones for a set of contours are visible. Formally,
a Venn diagram is of the shape \(d = (\contours, \vennzones{\contours}, \ezones)\).
Hence the only diagrammatic elements that may carry meaning are the presence of
contours, and whether a zone is shaded. For a given diagram \(d\),
we denote the set of shaded zones also by \(\ezonesOf{d}\).
We allow for the
diagrams \(\false = (\emptyset, \{(\emptyset, \emptyset)\}, \emptyset)\)
and \(\true = (\emptyset, \{(\emptyset, \emptyset)\}, \{(\emptyset, \emptyset)\})\).
A \emph{literal} is a Venn diagram for a single contour, with exactly one
shaded zone.
If the zone \((\emptyset, \{c\})\) is shaded in a literal, then we call it \emph{the negative literal  for \(c\)}, otherwise
it is \emph{the positive literal for \(c\)} (see Fig.~\ref{fig:literal-unitary}). Furthermore, if \(d\) is the positive literal for \(c\),
then we call the negative literal for \(c\) the \emph{dual  of \(d\)} (and vice versa). Observe that our notion of literals  deviates
from the original definition of Stapleton and Masthoff \cite{Stapleton2007} and from our previous work~\cite{Linker2018}. 
The main difference between our presentation
and classical Venn diagrams is the interpretation of shaded zones.
\setlength\intextsep{0pt}
\begin{wrapfigure}[4]{l}{.22\textwidth}
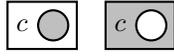

\centering
\unitaryLiteralPos{c}{.4cm} 
\unitaryLiteralNeg{c}{.4cm}
\caption{Literals}
  \label{fig:literal-unitary}
\end{wrapfigure}
\setlength{\intextsep}{\savedintextsep}
While in
the traditional approach, shading denotes the \emph{emptiness} of sets, we use shading
as a \emph{marker} of elements. That is, the semantics of a diagram
consists of the join of the elements denoted by the shaded zones. This
is more in line with the constructivist approach we want to emphasise: instead
of relying on a negative aspect (emptiness), we construct the semantics out of their
building blocks (the shaded zones).

\begin{definition}[Zone Semantics]
  \label{def:zone_sem}
  Let  \(\model\) be a Heyting algebra,  \(\valSing\) a  valuation, 
  and \(z\) a zone.
  The \emph{semantics} of 
  \(z\) is 
  given by
  \(\val{z} 
  = {\hbigand_{c \in \zinF{z}} \val{c}{}  \hand \hbigand_{c \in \zoutF{z}} {\hnot\val{c}{}}}
\).
\end{definition}
With the semantics of single zones defined, we can now define the semantics of a Venn diagram
in general.
\begin{definition}[Venn Diagram Semantics]
  \label{def:venn-semantics}
  For a Venn diagram \(d\), a Heyting algebra \(\model\) and a valuation \(\valSing\),
  the semantics of \(d\) are given by
  \( \val{d}  = \hbigor_{z \in \ezonesOf{d}} \val{z}\). 
\end{definition}

Note that we have \(\val{\top} = 1\) and \(\val{\bot} =  0\), for any Heyting algebra \(\model\) and valuation \(\valSing\). 
Furthermore, for a unitary diagram with a single contour and no shaded zones, i.e. \(d = (\{a\}, \vennzones{\{a\}},\emptyset)\),
we have \(\val{d} = 0\). However, the semantics already diverge from the classical case  for a fully shaded
diagram with one contour: if \(d = (\{a\}, \vennzones{\{a\}}, \vennzones{\{a\}})\), then
\(\val{d} = \val{a} \hor \hnot{\val{a}}\), which in general is not equal to \(1\).

Note that this semantics has one consequence in particular: we can decompose a zone into an equivalent
compound diagram, and we can furthermore decompose any unitary Venn diagram into a disjunctive normal form.
\begin{lemma}
  \label{lem:dnf}
  Let \(z \) be a zone for the contours \(L\). Then the semantics of the  compound diagram
  \(  d_z  = \bigwedge_{c \in \zinF{z}} \unitaryLiteralPos{c}{.2cm} \land \bigwedge_{c \in \zoutF{z}} \unitaryLiteralNeg{c}{.2cm}\)
  equals the semantics of \(z\), i.e. \(\val{d_z} = \val{z}\). Furthermore, for a Venn diagram \(d = (\contours, \vennzones{\contours}, \ezones)\), we have \(\val{d} = \val{\bigvee_{z \in \ezones} d_z}\).
\end{lemma}
\begin{proof}
  Immediate by the semantics in Def.~\ref{def:zone_sem} and Def.~\ref{def:venn-semantics}. \qed
\end{proof}
In particular, this implies that we cannot draw a unitary diagram that expresses intuitionistic implication.
\begin{lemma}
  Let \(a\) and \(b\) be propositional variables. Then there is no unitary Venn diagram \(d\) such that
  \(\val{d} = \val{a \implies b}\) for all models and valuations.
\end{lemma}
\begin{proof}
  By Lemma~\ref{lem:dnf}, every unitary diagram \(d\) can be expressed by using \(\lor\) and \(\land\) only.
  However, \(\implies\) is not definable by any combination of \(\lor\) and \(\land\)~\cite{Rasiowa1963}. \qed
\end{proof}
Observe however that we can trivially define a \emph{compound} diagram \(\unitaryLiteralPos{a}{.2cm} \implies \unitaryLiteralPos{b}{.2cm}\).

\paragraph{Pure Euler Diagrams}
\label{sec:euler-def}
We need additional syntax if we want to express intuitionistic
implication diagrammatically. This new syntax needs to be directed (since \(a \implies b\) is
semantically different to \(b \implies a\)). Observe that our notion of zones is \emph{already} directed, and
expresses topological information. So, a natural consideration is to allow for \emph{missing} zones in the diagrams. Hence,
instead of using Venn diagrams we will now discuss pure Euler diagrams. In contrast to shaded zones, we will treat
the missing zones as ``restrictions on the construction'' of the semantics. 
First, we give the semantics of a missing zone. 
\begin{definition}[Missing Zone Semantics]
  \label{def:mz-sematincs}
  For a Heyting algebra \(\model\),  a valuation \(\valSing\)
  and a zone \(z\), the \emph{missing zone semantics} of \(z\) is 
  given
by
\(\msem{z}
= \left(\hbigand_{c \in \zinF{z}} \val{c}\right) \himplies \left( \hbigor_{c \in \zoutF{z}} \val{c} \right)
\).
\end{definition}

\begin{definition}[Pure Euler Diagrams]
  \label{def:pure_euler}
  A \emph{pure Euler diagram}  is a structure \(d = (\contours, \zones)\),
  where \(\contours\) is the set of contours and \(\zones\) the set of visible
  zones of \(d\).  Furthermore, the set \(\mzones{d} = \vennzones{\contours} \setminus \zones\) is
  the set of \emph{missing zones} of \(d\). The semantics of pure Euler diagrams is that they require the constraints defined
by their missing zones to be true. That is, for a pure Euler diagram  \(d\), we have
\(  \val{d} = \hbigand_{z \in \mzones{d}} \msem{z}\). 
\end{definition}

In contrast to Venn diagrams, pure Euler diagrams do not
  allow for \emph{any} shading. To distinguish pure Euler diagrams from Venn diagrams (and
  Euler-Venn diagrams, see below), we draw them with
dotted contours.

Even with this additional syntax, we are not able to express \emph{every} implication.
A simple example would be \(a \implies a\), since we cannot have a zone \((\{a\}, \{a\})\).
However, for this particular example, we do not lose expressivity, since \(a \implies a \equiv \top\)
for all \(a\).
But we have a diagram equivalent to \(a \implies b\), 
\setlength\intextsep{0pt}
\begin{wrapfigure}[6]{l}{.35\textwidth}
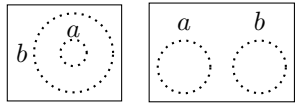

\centering
\subsetEDCstr{b}{a}{.35cm}
\disjointEDCstr{a}{b}{.7cm}
\caption{Pure Euler Diagrams}
  \label{fig:imp-disj}
\end{wrapfigure}
\setlength{\intextsep}{\savedintextsep}
as shown in the left diagram of
Fig.~\ref{fig:imp-disj}.
The right diagram in Fig.~\ref{fig:imp-disj} denotes \((a \hand b) \himplies 0\), which
is \(-(a \hand b)\).
Observe that in contrast to Venn diagrams without shaded zones,
a pure Euler diagram without   missing zones denotes \(1\), i.e.,  for \(d = (\contours, \vennzones{\contours})\), 
we have \(\val{d} = \val{\top} = 1\). Furthermore, the diagram without any contours
and zones denotes \(0\), since
\(\val{(\emptyset, \emptyset)} = \msem{(\emptyset, \emptyset)} =
\hbigand_{c \in \emptyset} \val{c} \himplies \hbigor_{c \in \emptyset} \val{c} = 1 \himplies 0 = 0\). 
In the following, we will need to identify zones that are divided by a contour \(c\) abstractly.
Intuitively such a zone is split into two zones \(z\) and \(z^\prime\) that only differ insofar, as \(c\)
is in \(\zinF{z}\) and in \(\zoutF{z^\prime}\).

\begin{definition}[Adjacent Zone]
Let \(z = (\zin, \zout)\) be a zone for the contours in
\(\contours\) and \(c \in \contours\). The \emph{zone adjacent to 
\(z\) at \(c\)}, denoted by \(\adj{z}{c}\) is \((\zin \cup \{c\}, \zout \setminus \{c\})\), if \(c \in \zout\)
and \((\zin \setminus \{c\}, \zout \cup \{c\})\) if \(c \in \zin\). 
\end{definition}

Now we can define a way to 
remove contours from a  pure Euler diagram \(d\). This contrasts to our previous work, where we
allowed that the diagram to be reduced contains shading \cite{Linker2018}. 
\begin{definition}[Reduction]
\label{def:reduction}
Let \(d = (\contours, \zones)\) be a  pure Euler diagram and
\(c \in \contours\). The \emph{reduction} of a zone 
\(z = (\zin, \zout)\) is 
\(  z \setminus c = (\zin \setminus \{c\}, \zout \setminus \{c\})\).
The \emph{reduction} of \(d\) by \(c\) is defined as \(d \setminus c = (\contours \setminus \{c\}, \zones\setminus c)\),
where
\(\zones \setminus c  = \{ z \setminus c \mid z \in \zones\}\).
\end{definition}

\begin{lemma}[Properties of Reduction]
  \label{lem:prop-reduction}
We have \(z \setminus c = \adj{z}{c} \setminus c\).
Furthermore, for each \(z^\prime \in \mzones{d\setminus c }\) and \(z\) with
\(z \setminus c = z^\prime\), we have \(z \in \mzones{d}\). In particular,
both \(z \in \mzones{d}\) and \(\adj{z}{c} \in \mzones{d}\).
\end{lemma}
\begin{proof} Immediately from the definition of reduction.\qed\end{proof}
 If each
 missing zone in a pure Euler diagram \(d\) has a  missing adjacent zone,
 then the  reduction of \(d\) by any contour is contained in the semantics of \(d\). In
 particular, the meet of all reductions equals the semantics of \(d\). This will
 allow us to show soundness of some rules of the sequent calculus in Sect.~\ref{sec:sequent}.
 \begin{lemma}
\label{lem:sem-reduction}
Let \(d = (\contours, \zones)\) be a pure Euler diagram,  
where 
for each \(z \in  \mzones{d}\), there is a contour \(\ell \in \contours\) such that 
\(\adj{z}{\ell} \in  \mzones{d}\). Furthermore, let \(\contours^\prime = \{ c \mid \mzones{d \setminus c} \neq \emptyset\}\).
Then \(\hbigand_{c \in \contours^\prime} \val{d \setminus c} = \val{d} \)
\end{lemma}
\begin{proof}
  Let \( c \in \contours^\prime\) and \(z^\prime = (\zin, \zout) \in \mzones{d \setminus c}\).
  Then, let \(z =(\zin \cup\{c\}, \zout)\). That is \(z \setminus c = z^\prime\) and  \(z, \adj{z}{c} \in \mzones{d}\)
  by Lemma~\ref{lem:prop-reduction} (if \(z = (\zin, \zout \cup \{c\})\),
  we can reverse the roles of \(z\) and \(\adj{z}{c}\) in the following).
  Assume \(x \leq \msem{z \setminus c}\)
  Then we have \(x \leq \hbigand_{a \in \zin} \val{a} \himplies  \hbigor_{a \in \zout} \val{a}\), if,
  and only if, \(x \hand \hbigand_{a \in \zin} \val{a}  \leq \hbigor_{a \in \zout} \val{a}\). This implies \(x \hand \hbigand_{a \in \zin} \val{a}  \leq \hbigor_{a \in \zout} \val{a} \hor \val{c}\),
  which is equivalent to \( x \leq \hbigand_{a \in \zin} \val{a}  \himplies \hbigor_{a \in \zout} \val{a} \hor \val{c} = \msem{\adj{z}{c}}\). Also, from \(x \leq \msem{z \setminus c}\), we have
  \(x \hand \hbigand_{a \in \zin} \val{a} \hand \val{c}  \leq  \hbigand_{a \in \zin} \val{a}  \leq \hbigor_{a \in \zout} \val{a}\), which gives us
  \(x \leq \hbigand_{a \in \zin} \val{a} \hand \val{c} \himplies \hbigor_{a \in \zout} \val{a} = \msem{z}\). Hence, we have \(x \leq \msem{z} \hand \msem{\adj{z}{c}}\). That is,
  for each \(z^\prime \in \mzones{d\setminus c}\), we have a \(z \in \mzones{d}\) such that 
  \(\msem{z^\prime} = \msem{z \setminus c} \leq \msem{z} \hand \msem{\adj{z}{c}}\). Thus, we have \(\val{d \setminus c} \leq \val{d}\) for each
  \(c \in \contours^\prime\), that is \(\hbigand_{c \in \contours^\prime} \val{d \setminus c} \leq \val{d}\).

 Conversely, let \(x \leq \val{d}\), i.e.
 \( x \leq \hbigand_{z \in \mzones{d}} \left ( \hbigand_{a \in \zinF{z}} \val{a} \himplies \hbigor_{a \in \zoutF{z}} \val{a}\right)\).
  For an arbitrary \(z \in \mzones{d}\), choose \(c \in \zinF{z}\) and \(c \in L^\prime\), i.e.,
  the zone \(z \setminus c\) is missing in at least one diagram (namely \(d \setminus c\)).
  Of course, we have
  \(x \leq \msem{z}\), from which we get by Lemma~\ref{lem:prop}~(\ref{eq:impl_and})
  \(x \leq \hbigand_{a \in \zinF{z}\setminus \{c\}} \val{a} \himplies \left( \val{c} \himplies \hbigor_{a \in \zoutF{z}} \val{a} \right )\), which is equivalent to
\(    x \hand \hbigand_{a \in \zinF{z}\setminus \{c\}} \val{a} \leq  \val{c} \himplies \hbigor_{a \in \zoutF{z}} \val{a} \). 
  Furthermore, from \(x \leq \msem{\adj{z}{c}}\), we also have
 \(   x \hand \hbigand_{a \in \zinF{z}\setminus \{c\}} \val{a} \leq  \val{c} \hor \hbigor_{a \in \zoutF{z}} \val{a}\). 
  By the properties of a distributive lattice, and Lemma~\ref{lem:prop}~(\ref{eq:mp}) and (\ref{eq:imp_b}), we then get
\(
    x \hand \hbigand_{a \in \zinF{z}\setminus \{c\}} \val{a} \leq  \left(\val{c} \hor \hbigor_{a \in \zoutF{z}} \val{a}\right) \hand \left(  \val{c} \himplies \hbigor_{a \in \zoutF{z}} \val{a}\right )
                                                              \leq \left( \val{c} \hand \left(  \val{c} \himplies \hbigor_{a \in \zoutF{z}} \val{a}\right )\right)
                                                               \hor \left ( \hbigor_{a \in \zoutF{z}} \val{a} \hand \left(  \val{c} \himplies \hbigor_{a \in \zoutF{z}} \val{a}\right )\right)
     \leq \hbigor_{a \in \zoutF{z}} \val{a}\),
     which  is equivalent to \(x \leq \hbigand_{a \in \zinF{z} \setminus \{c\}} \val{a} \himplies \hbigor_{a \in \zoutF{z}}\val{a} = \msem{z \setminus c}\).
     Now, since \(z\) was arbitrary, this reasoning
  holds for all \(z \in \mzones{d}\) (possibly with the roles of \(z\) and \(\adj{z}{c}\) reversed),
  and thus \(x \leq \hbigand_{c \in L^\prime} \val{d \setminus c}\), and hence \(\val{d} \leq \hbigand_{c \in L^\prime} \val{d \setminus c}\). \qed
\end{proof}
For an example, consider the  derivation in Sect.~\ref{sec:admissible}. The diagram \(d_C^\ast\) as
shown in Table~\ref{tab:example_abbrev} 
can be reduced to the three diagrams shown in the application of rule \(\reduceL\) in derivation \(\Pi_1\)
presented in Fig.~\ref{fig:aux_deriv}.

\paragraph{Euler-Venn Diagrams}
\label{sec:euler-venn-def}
In this section, we combine pure Euler diagrams with the central
syntactic aspect of Venn diagrams: shading. Our main idea can be summarised
as follows: We treat the information given by a pure Euler diagram
as a condition for the construction of the  combinations of atomic
propositions denoted by the shading. That is, whenever we have constructions
as indicated by the spatial relations of contours in a diagram \(d\), we also
have a construction of the elements denoted by the shaded zones of the diagram.
Since we use the syntactic elements of pure Euler diagrams and Venn diagrams,
we will subsequently call such diagrams \emph{Euler-Venn diagrams}.

The abstract syntax of Euler-Venn diagrams is similar to Venn diagrams.
A diagram is a tuple \(d = (\contours, \zones, \ezones)\) consisting of
a set of contours \(\contours\), a set of visible zones \(\zones\) over \(\contours\), and
a set of shaded zones \(\ezones \subseteq \zones\). We will often need
to refer to the pure Euler or Venn aspects of an Euler-Venn diagram separately. Hence, we introduce
some additional notation. For an Euler-Venn diagram
 \(d = (\contours, \zones, \ezones)\) we will write
 \(\venn{d} =  (\contours, \vennzones{\contours}, \ezones )\) for
 the Venn diagram with the same set of shaded zones as \(d\),
 and \(\euler{d} =(\contours, \zones)\) for the pure Euler diagram with the same set of
 visible zones as \(d\). Similarly to pure Venn and Euler diagrams, we will
 refer to the missing zones of \(d\) by  \(\mzones{d}\) and
 to its shaded zones by \(\ezonesOf{d}\). 

\begin{definition}[Euler-Venn Diagram Semantics]
  \label{def:euler-venn-semantics}
  The  semantics of a unitary \emph{Euler-Venn diagram} for a Heyting algebra \(\model\) and a valuation
  \(\valSing\) is \(\val{d}  = 
    \val{\euler{d}} \himplies \val{\venn{d}}\).

 \end{definition}
    Observe that  with this definition, the semantics for the case \(\mzones{d} = \emptyset\) and \(\ezonesOf{d} \neq \emptyset\) yields
 \(\val{d} = 1 \himplies \hbigor_{z \in \ezonesOf{d}} \val{z} = \hbigor_{z \in \ezonesOf{d}} \val{z}\).
 Furthermore, we get \(\val{\false} = 1 \himplies 0 = 0\) and
 \(\val{\true} = 1 \himplies 1 = 1\).

Observe that the language of compound Euler-Venn diagrams can be seen as a subset of
intuitionistic logic. In particular, we can translate every diagram into a formula, which
we call its \emph{canonical formula}.

\begin{definition}[Canonical Formula]
  \label{def:canonical}
  The \emph{canonical formula} of an   Euler-Venn diagram is given by the following
  recursive definition. We start with the definition of the canonical formula of shaded
  and missing zones. 
  \begin{align*}
    \canon{z}{z} &= \bigland_{c \in \zinF{z}} c \land \bigland_{c \in \zoutF{z}} -c & 
    \canon{z}{m} &= \bigland_{c \in \zinF{z}}  c \implies \biglor_{c \in \zoutF{z}} c 
  \end{align*}
  For a pure Euler diagram \(d_e\), a Venn diagram \(d_v\), an Euler-Venn diagram
  \(d\) and compound diagrams \(D\) and \(E\),  the canonical formula is given as
  \begin{align*}
    \canon{d_e}{} & = \bigland_{z \in \mzones{d_e}} \canon{z}{m}&
    \canon{d_v}{} & = \biglor_{z \in \ezonesOf{d_v}} \canon{z}{z}\\
    \canon{d}{} & = \canon{\euler{d}}{} \implies \canon{\venn{d}}{} &
    \canon{D \otimes E}{} & = \canon{D}{} \otimes \canon{E}{}\enspace,  \otimes \in \{\land, \lor, \implies \}
  \end{align*}
\end{definition}

\begin{remark}
  \label{rem:simplified-canonical-literals}
  Observe that according to Def.~\ref{def:canonical}, we get
  \(\canon{\unitaryLiteralPos{c}{.15cm}}{} = c \land \true\)
  and \(\canon{\unitaryLiteralNeg{c}{.15cm}}{} = \true \land -c\). However,
  for simplicity, we will assume that the canonical formula construction
  omits superfluous occurences of \(\true\) and \(\false\). Hence,
  \(\canon{\unitaryLiteralPos{c}{.15cm}}{} = c\) and
  \(\canon{\unitaryLiteralNeg{c}{.15cm}}{} = -c\). Similarly, e.g.,
  \(\canon{(\emptyset, \contours)}{m} = \biglor_{c \in \contours} c\). 
\end{remark}


\section{Sequent Calculus}
\label{sec:sequent}
Sequent calculus, as defined by Gentzen \cite{Gentzen1935} is closely
related to natural deduction. It is based on \emph{sequents}, which are decomposed by rule applications.
In the following, we will define a multi-succedent version of sequent calculus
for Euler-Venn diagrams called \(\diagName\). This version is inspired by the work of
Dragalin \cite{Dragalin1988}, while following the more modern presentation
of Negri et al. \cite{Negri2001}. 
\begin{definition}[Sequent]
\label{def:sequent}
A \emph{sequent} \(\sequent{\leftSeq}{\rightSeq}\) consists
of multisets \(\leftSeq\) and \(\rightSeq\) of 
Euler diagrams. The multiset \(\leftSeq\) is called
the \emph{antecedent} and \(\rightSeq\) the \emph{succedent}. 

If \(\leftSeq\) (\(\rightSeq\)) is the empty multiset, we write
\(\sequent{}{\rightSeq}\) (\(\sequent{\leftSeq}{}\), respectively). 
If a sequent is of the form \(\sequent{p,\leftSeq}{\rightSeq, p}\) 
where \(p\) is a positive literal,  
then it is called
an \emph{axiom}. 
A sequent
\(\sequent{D_1 , \dots,  D_k}{E_1, \dots, E_{l}}\) is valid, if, and only if,
\(\val{D_1}  \hand \dots \hand \val{D_k} \leq \val{E_1} \hor \dots \hor \val{E_l}\) for all
valuations \(\valSing\) in all
Heyting algebras. We will often abbreviate \(\val{D_1}  \hand \dots \hand \val{D_k}\)
by \(\val{\leftSeq}\) and \(\val{E_1} \hor \dots \hor \val{E_l}\) by \(\val{\rightSeq}\).
That is, for the 
multiset \(\leftSeq\) we always mean the meet, while for \(\rightSeq\) we always
refer to the join of the diagrams it consists of.
\end{definition}

A \emph{deduction} for a sequent \(\sequent{\leftSeq}{\rightSeq}\)
is a tree, where the root is labelled by \(\sequent{\leftSeq}{\rightSeq}\),
and the children of each node are labelled according to the rules defined 
below. If the validity of the premisses of a rule imply the validity of its
conclusion, we call the rule \emph{sound}.
A deduction where the leaves are labelled with
axioms, or instances of \(\falseL\) and \(\trueR\), is called a \emph{proof} for \(\sequent{\leftSeq}{\rightSeq}\).
We will write \(\provable \sequent{\leftSeq}{\rightSeq}\) to denote the
existence of a proof for \(\sequent{\leftSeq}{\rightSeq}\).
In all rules, we call the diagram in the conclusion that is being decomposed the
\emph{principal diagram} of the rule. For example, in \(\landL\), the principal
diagram is \(D \land E\), and in the rule \(\sepL\) it is \(d\).
For a given proof of \( \sequent{\leftSeq}{\rightSeq}\), its \emph{height}
is the highest number of successive proof rule applications~\cite{Negri2001}.
We will write \(\provable_n \sequent{\leftSeq}{\rightSeq}\) if \(\sequent{\leftSeq}{\rightSeq}\) is provable with a proof
of height at most \(n\).

We now turn to define and explain the rules of \(\diagName\).
The rules to treat compound diagrams, as shown in Fig.~\ref{fig:rules-boolean}, are directly taken from 
sequent calculus
for intuitionistic propositional logic and are sound.

\begin{lemma}[Soundness]
\label{lem:rules-boolean-sound}
The rules for sentential operators are sound.  
\end{lemma}
\begin{proof}
  A straightforward adaptation of the proofs shown by Ono~\cite{Ono2019}. \qed
\end{proof}

\begin{remark}
  \label{rem:props-sentential-calculus}
  If we take the placeholders  \(D\), \(E\) and \(F\) as formulas according
  to Def.~\ref{def:logic_syn} and both \(\leftSeq\) and \(\rightSeq\) as
  multisets  of such formulas, then the rules 
  of Fig.~\ref{fig:rules-boolean} together with axioms  \(\sequent{p, \leftSeq}{\rightSeq,p}\)
  form the sentential sequent calculus \(\sentName\)~\cite{Negri2001}.
 Provability in
  \(\sentName\) is equivalent to provability in
  Gentzen's  system \(\mathsf{LJ}\).
  The system \(\mathsf{LJ}\) is sound and complete~\cite{Ono2019}. Hence,
  \(\sentName\) is sound and complete as well. Furthermore, the structural
  rules of weakening, contraction and cut are admissible~\cite{Negri2001}. Observe
  that we treat \(\falseL\) as a \emph{rule}, and not as an axiom.
\end{remark}

\begin{figure}[t]
  \centering
  \begin{tabular}{ccc}
    \prftree[r]{\(\scriptstyle \landL\)}
    {\sequent{D,E,\leftSeq}{\rightSeq}}
    {\sequent{D \land E, \leftSeq}{\rightSeq}}
    &
      \prftree[r]{\(\scriptstyle \lorL\)}
      {\sequent{D,\leftSeq}{\rightSeq}}
      {\sequent{E,\leftSeq}{\rightSeq}}
      {\sequent{D \lor E,\leftSeq}{\rightSeq}}
    &
      \prftree[r]{\(\scriptstyle \impliesL\)}
      {\sequent{\leftSeq, D \implies E}{ D}}
      {\sequent{E, \leftSeq}{\rightSeq}}
      {\sequent{ D \implies E,\leftSeq}{\rightSeq}}
    \\[1em]
    \prftree[r]{\(\scriptstyle \landR\)}
    {\sequent{\leftSeq}{\rightSeq,D}}
    {\sequent{\leftSeq}{\rightSeq,E}}
    {\sequent{\leftSeq}{ \rightSeq, D \land E }}
    &
      \prftree[r]{\(\scriptstyle \lorR\)}
      {\sequent{\leftSeq}{\rightSeq, D,E}}
      {\sequent{\leftSeq}{\rightSeq, D \lor E}}
    &
      \prftree[r]{\(\scriptstyle \impliesR\)}
      {\sequent{D,\leftSeq}{E}}
      {\sequent{\leftSeq}{\rightSeq, D \implies E}}\\[1em]
    &
      \prftree[r]{\(\scriptstyle \falseL\)}{\sequent{\leftSeq, \bot}{\rightSeq}}
    &
  \end{tabular}
  \caption{Proof Rules for Sentential Operators}
  \label{fig:rules-boolean}
\end{figure}

\paragraph{Rules for Venn Diagrams.}
The rules in \ref{fig:rules-venn-literal} let us reduce negative
to positive literals. Observe that we may introduce arbitrary
sets of formulas into the succedent.  This ensures admissability of
the structural rules (cf. Lemma~\ref{lem:weakening} and~\ref{lem:contraction}).
Furthermore, the rule \(\trueR\)
lets us finish a proof similarly to \(\falseL\).
Let \(d = (\contours, \vennzones{\contours},\ezones)\) be a Venn diagram with
\(|\ezones| > 1\), and let 
\(d_i =(\contours, \vennzones{\contours}, \ezones_i)\), for \(i \in \{1,2\}\),
such that
\(\ezones = \ezones_1 \cup \ezones_2\). Then the rules \(\sepL\) and \(\sepR\)
in Fig.~\ref{fig:rules-venn-separation}
\emph{separate} \(d\) into \(d_1\) and \(d_2\). These rules are closely
related to the \emph{Combine}
equivalence rule for Spider diagrams \cite{Howse2005}.
For a Venn diagram \(d\) with \(\ezonesOf{d} = \{z\}\), where
\(z = (\{n_1, \dots, n_k\}, \{o_1, \dots, o_l\})\), the rules \(\singDecL\)
and \(\singDecR\) of Fig.~\ref{fig:rules-venn-decompose} \emph{decompose}
the single zone \(z\) into literals.

 \begin{figure}[h]
\begin{center}
\subfloat[]{
       \prftree[r]{\(\scriptstyle \litL\)}
      {\sequent{\unitaryLiteralNeg{c}{.15cm}, \leftSeq}{\unitaryLiteralPos{c}{.15cm}}}
      {\sequent{\unitaryLiteralNeg{c}{.15cm}, \leftSeq}{\rightSeq}}
       \hspace{1em}
      \prftree[r]{\(\scriptstyle \litR\)}
      {\sequent{\unitaryLiteralPos{c}{.15cm},\leftSeq }{}}
      {\sequent{\leftSeq}{\rightSeq, \unitaryLiteralNeg{c}{.15cm}}}
       \hspace{1em}
       \prftree[r]{\(\scriptstyle \trueR\)}
       {\sequent{\leftSeq}{\rightSeq, \unitaryTrue{.05cm}}}

      \label{fig:rules-venn-literal}
    }

  \subfloat[]{
\begin{tabular}{c}
  \prftree[r]{\(\scriptstyle \sepL\)}
{\sequent{d_1,\leftSeq}{\rightSeq}}
{\sequent{d_2,\leftSeq}{\rightSeq}}
  {\sequent{ d,\leftSeq}{\rightSeq}}
  \hspace{1em}
\prftree[r]{\(\scriptstyle \sepR\)}
{\sequent{\leftSeq}{\rightSeq, d_1, d_2}}
{\sequent{\leftSeq}{\rightSeq, d}}
\end{tabular}
\label{fig:rules-venn-separation}
}

\subfloat[]{
    \begin{tabular}{c}    
      \prftree[r]{\(\scriptstyle \singDecL\)}
      {\sequent{ \unitaryLiteralPos{n_1}{.15cm}, \dots, \unitaryLiteralPos{n_k}{.15cm}, \unitaryLiteralNeg{o_1}{.15cm}, \dots, \unitaryLiteralNeg{o_l}{.15cm},\leftSeq}{\rightSeq }}
      {\sequent{d,\leftSeq}{\rightSeq}}
      \\[1em]
      \prftree[r]{\(\scriptstyle \singDecR\)}
      {\sequent{ \leftSeq}{\rightSeq, \unitaryLiteralPos{n_1}{.15cm}}}
            {\dots}
      {\sequent{ \leftSeq}{\rightSeq, \unitaryLiteralPos{n_k}{.15cm}}}
      {\sequent{   \leftSeq}{\rightSeq, \unitaryLiteralNeg{o_1}{.15cm} }}
      {\dots}
      {\sequent{\leftSeq}{\rightSeq, \unitaryLiteralNeg{o_l}{.15cm} }}
      {\sequent{\leftSeq}{\rightSeq, d}}\\[1em]
    \end{tabular}
\label{fig:rules-venn-decompose}
}

  \end{center}
\caption{Rules for Unitary Venn Diagrams}
\label{fig:rules-venn}
\end{figure}
\begin{lemma}
  \label{lem:rules-venn-sound}
  The rules shown in Fig.~\ref{fig:rules-venn} are sound.
\end{lemma}
\begin{proof}
  In all of the following cases, let \(\valSing\) be an arbitrary valuation. The rule \(\trueR\) is clearly
  sound, since \(\val{\true} = 1\) for any valuation.
  For \(\litL\), assume \(\val{\unitaryLiteralNeg{c}{.15cm}} \hand \val{\leftSeq} \leq \val{\unitaryLiteralPos{c}{.15cm}}\).
  Then, we have \(\val{\unitaryLiteralNeg{c}{.15cm}} \hand \val{\leftSeq} = \val{\unitaryLiteralNeg{c}{.15cm}} \hand \val{\leftSeq}
  \hand \val{\unitaryLiteralNeg{c}{.15cm}} \hand \val{\leftSeq} \leq \val{\unitaryLiteralNeg{c}{.15cm}} \hand \val{\leftSeq} \hand
  \val{\unitaryLiteralPos{c}{.15cm}} \leq \val{\leftSeq} \hand 0 = 0 \leq \val{\rightSeq}\), where
  the first inequality is an application of the assumption, and the second is due to Lemma~\ref{lem:prop}~(\ref{eq:mp}).
  For \(\litR\), assume \(\val{\unitaryLiteralPos{c}{.15cm}} \hand \val{\leftSeq} \leq 0\). Then we
  get, by the definition of the implication, the lattice properties, and the semantics of literals,
  \(\val{\leftSeq} \leq \val{\unitaryLiteralPos{c}{.15cm}} \himplies 0  = \val{\unitaryLiteralNeg{c}{.15cm}} \leq \val{\rightSeq} \hor \val{\unitaryLiteralNeg{c}{.15cm}}\). 
  
  Consider \(\sepR\). Assume \(\val{\leftSeq} \leq \val{\rightSeq} \hor \val{d_1} \hor \val{d_2} \),
  we have in particular \(\val{\leftSeq} \leq \val{\rightSeq} \hor \hbigor_{z \in \ezones_1} \val{z}
  \hor \hbigor_{z \in \ezones_2} \val{z}\). Since
  \(\ezones_1  \cup \ezones_2  = \ezones\), and since we can ignore duplicate
  contour semantics by the lattice properties of Heyting algebras, 
  \(\val{\leftSeq} \leq \val{\rightSeq} \hor \hbigor_{z \in \ezones} \val{z}\), i.e., \(\val{\leftSeq} \leq \val{\rightSeq} \hor \val{d}\).
  Now consider \(\sepL\).
  We have both \(\val{d_1} \hand \val{\leftSeq} \leq \val{\rightSeq}\) and
  \(\val{d_2} \hand \val{\leftSeq} \leq \val{\rightSeq}\), i.e.,
  \begin{align*} 
  &(\hbigor_{z \in \ezones_1} \val{z} \hand \val{\leftSeq}) \hor (\hbigor_{z \in \ezones_2}\val{z} \hand \val{\leftSeq}) &&\leq \val{\rightSeq} \hor \val{\rightSeq}\\
\iff& ( \hbigor_{z \in \ezones_1} \val{z}  \hor \hbigor_{z \in \ezones_2} \val{z}) \hand \val{\leftSeq}
                                                                                                              &&\leq \val{\rightSeq}\\
    \iff & (\hbigor_{z \in \ezones} \val{z})  \hand \val{\leftSeq} && \leq{\val{\rightSeq}} 
  \end{align*}
which is exactly \(\val{d} \hand   \val{\leftSeq} \leq \val{\rightSeq}\).

Now consider \(\singDecL\). By  Def.~\ref{def:sequent}, the premiss denotes
\(\val{n_1} \hand \dots \hand \val{n_k} \hand -\val{o_1} \hand \dots \hand -\val{o_l} \hand \val{\leftSeq} \leq \val{\rightSeq}\).
But since \(z\) is the only shaded zone of \(d\),
this is exactly the semantics of \(\sequent{d, \leftSeq}{\rightSeq}\), by
Def.~\ref{def:zone_sem} and Def.~\ref{def:sequent}.
Finally, consider \(\singDecR\). Then, we have
\(\val{\leftSeq} \leq \val{\rightSeq} \hor \val{n_i}\) and
\(\val{\leftSeq} \leq \val{\rightSeq} \hor  -\val{o_j}\) for all
\(i \in \{1, \dots , k\}\) and \(j \in \{1, \dots, l\}\). By the
lattice properties, we get
\(\val{\leftSeq} \leq (\val{\rightSeq} \hor \val{n_1})\hand \dots \hand
(\val{\rightSeq} \hor \val{n_k}) \hand (\val{\rightSeq} \hor -\val{o_1})
\hand \dots \hand (\val{\rightSeq} \hor  -\val{o_l})\), which is, by
distributivity and since \(z\) is the only shaded zone in \(d\), the same as
\(\val{\leftSeq} \leq \val{\rightSeq} \hor \val{d}\). \qed
\end{proof}

\paragraph{Rules for pure Euler Diagrams.}
Now let \(d = (\contours, \zones)\) be a pure Euler diagram, where for each \(z \in \mzones{d}\) there
is a contour \(\ell \in \contours\), such that \(\adj{z}{\ell} \in \mzones{d}\).  Furthermore, let  \( \{c_1, \dots, c_k\} \subseteq \contours\)
 be the maximal set of contours such that \(\mzones{d \setminus c_i} \neq \emptyset\) for every \(i \leq k\).
 Then we can \emph{reduce} \(d\) according to the rules \(\reduceL\) and \(\reduceR\) shown in Fig.~\ref{fig:rules-decompose-missing}.
Let \(d= (\contours, \zones)\) be a pure Euler diagram with more than one missing zone, i.e., \(|\mzones{d}| > 1\),
and let \(d_1 = (\contours, \zones_1)\) and \(d_2= (\contours, \zones_2)\) be two pure Euler diagrams such that \(\zones_1 \cap \zones_2 = \zones\). Then
the rules \(\mzSepL\) and \(\mzSepR\) of Fig.~\ref{fig:rules-separate-imp} \emph{separate} the diagram \(z\) at its missing zones.
If \(d\) is a pure Euler diagram with a single missing zone, i.e. \(\mzones{d} = \{z\}\) and
\(z = ( \{n_1, \dots, n_k\},\{o_1, \dots, o_\ell\})\), then the rules of Fig.~\ref{fig:rules-decompose-imp}
decompose \(z\) into literals. 
\begin{figure}[h]
  \begin{center}
    \subfloat[]{
      \begin{tabular}{c}    
        \prftree[r]{\(\scriptstyle \reduceL\)}
        {\sequent{ d\setminus c_1, \dots, d\setminus c_k, \leftSeq}{\rightSeq}}
        {\sequent{d, \leftSeq}{\rightSeq}}
        \hspace{1em}
        \prftree[r]{\(\scriptstyle \reduceR\)}
        {\sequent{\leftSeq}{\rightSeq, d \setminus c_1}}
        {\!\dots\!}
        {\sequent{\leftSeq}{\rightSeq, d \setminus c_k}}  
        {\sequent{\leftSeq}{\rightSeq, d}}
      \end{tabular}
      \label{fig:rules-decompose-missing}
    }
    
    \subfloat[]{
      \begin{tabular}{c}
        \prftree[r]{\(\scriptstyle \mzSepL\)}
        {\sequent{d_1, d_2,\leftSeq}{\rightSeq}}
        {\sequent{ d,\leftSeq}{\rightSeq}}
        \hspace{1em}
        \prftree[r]{\(\scriptstyle \mzSepR\)}
        {\sequent{\leftSeq}{\rightSeq, d_1}}
        {\sequent{\leftSeq}{\rightSeq, d_2}}
        {\sequent{\leftSeq}{\rightSeq, d}}
      \end{tabular}
      \label{fig:rules-separate-imp}
    }
    
    \subfloat[]{
      \begin{tabular}{c}
        \prftree[r]{\(\scriptstyle \impDecL\) }
        {\sequent{d,\leftSeq}{\unitaryLiteralPos{n_1}{.15cm}}}
        {\dots}
        {\sequent{d,\leftSeq}{\unitaryLiteralPos{n_k}{.15cm}}}
        {\sequent{\unitaryLiteralPos{o_1}{.15cm},  \leftSeq}{\rightSeq}}
        {\dots}
        {\sequent{\unitaryLiteralPos{o_l}{.15cm},  \leftSeq }{\rightSeq}}
        {\sequent{d, \leftSeq}{\rightSeq}}\\[1em]
        \prftree[r]{\(\scriptstyle \impDecR\)}
        {\sequent{\leftSeq,\unitaryLiteralPos{n_1}{.15cm}, \dots, \unitaryLiteralPos{n_k}{.15cm}}{ \unitaryLiteralPos{o_1}{.15cm},\dots,\unitaryLiteralPos{o_l}{.15cm} }}
        {\sequent{\leftSeq}{\rightSeq,d}}
      \end{tabular}
      \label{fig:rules-decompose-imp}
    }
  \end{center}
  \caption{Proof Rules for pure Euler Diagrams}
  \label{fig:rules-euler-diagrams}
\end{figure}

\begin{lemma}
\label{lem:rules-euler-sound}
The rules shown in Fig.~\ref{fig:rules-euler-diagrams} are sound.
\end{lemma}
\begin{proof}
  The soundness of the rules \(\reduceL\) and \(\reduceR\) is immediate by Lemma~\ref{lem:sem-reduction}.  For rules \(\mzSepL\) and \(\mzSepR\) observe that by the condition on \(d_1\) and
  \(d_2\), we have \(\mzones{d_1} \cup \mzones{d_2} = \mzones{d}\). That is, \(\val{d_1} \hand \val{d_2} = \val{d}\) for
  all valuations and Heyting algebras. The soundness of both \(\mzSepL\) and \(\mzSepR\) follows
  by straightforward computations.
For the rule \(\impDecR\), the proof is straightforward by the definition
  of \(\himplies\) and the lattice properties.  The rule \(\impDecL\) can be proven sound
  similarly to  \(\eqDecL\). \qed
\end{proof}

\paragraph{Rules for Euler-Venn Diagrams.}
Let \(d\) be  an Euler-Venn diagram.
Then the rules \(\detL\)
 and \(\detR\) of Fig.~\ref{fig:rules-euler-venn-diagrams} \emph{detach}
 the  spatial relations from the shading.

\begin{figure}[h]
\begin{center}
  \begin{tabular}{cc}    
    \prftree[r]{\(\scriptstyle \detL\)}
    {\sequent{d, \leftSeq}{\euler{d}}}
    {\sequent{\venn{d},\leftSeq}{\rightSeq }}
    {\sequent{d,\leftSeq}{\rightSeq}} &\hspace{2em}
    \prftree[r]{\(\scriptstyle \detR\)}
    {\sequent{\euler{d}, \leftSeq}{ \venn{d}}}
    {\sequent{\leftSeq}{\rightSeq, d}}
  \end{tabular}
\end{center}
\caption{Proof Rules For Euler-Venn Diagrams}
\label{fig:rules-euler-venn-diagrams}
\end{figure}

\begin{lemma}
  \label{lem:rules-euler-venn-sound}
  The rules shown in  Fig.~\ref{fig:rules-euler-venn-diagrams} are sound.
\end{lemma}
\begin{proof}
  Consider \(\detR\), and assume \(\val{\euler{d}} \hand \val{\leftSeq} \leq \val{\venn{d}}\). Then,
  by Def.~\ref{def:heyting}, this is equivalent to \(\val{\leftSeq} \leq \val{\euler{d}} \himplies \val{\venn{d}}\),
  which by   Def.~\ref{def:euler-venn-semantics} and the lattice properties 
  implies \(\val{\leftSeq} \leq  \val{\rightSeq} \hor \val{d}\). So consider \(\detL\), and assume both
  \(\val{d} \hand \val{\leftSeq} \leq \val{\euler{d}}\) and \(\val{\venn{d}} \hand \val{\leftSeq} \leq \val{\rightSeq}\).
  We then have \(\val{d} \hand \val{\leftSeq} =
  \val{d} \hand \val{\leftSeq} \hand \val{d} \hand \val{\leftSeq} \leq
  \val{d} \hand \val{\leftSeq} \hand \val{\euler{d}} \leq
  \val{\venn{d}} \hand \val{\leftSeq} \leq \val{\rightSeq}\).
  The inequalities are correct  due to the first premiss,  Lemma~\ref{lem:prop}~(\ref{eq:mp}) and the second premiss, respectively. \qed
\end{proof}

By an induction on the height of proofs,   we 
get the 
soundness theorem for \(\diagName\), using Lemma~\ref{lem:rules-boolean-sound}, \ref{lem:rules-venn-sound}, \ref{lem:rules-euler-sound}, and~\ref{lem:rules-euler-venn-sound}.

\begin{theorem}[Soundness]
  \label{thm:soundess}
  If \(\sequent{\leftSeq}{\rightSeq}\) is provable in \(\diagName\), then \(\sequent{\leftSeq}{\rightSeq}\) is valid.
\end{theorem}

To prove completeness of the system, we first show that certain rules
are invertible. Even stronger, 
a rule is \emph{height-preserving invertible}, if whenever we have a proof of
height \(n\) for its conclusion, its premisses are provable with a proof of at most height \(n\).

\begin{lemma}[Inversions]
  \label{lem:inversions}
 \begin{enumerate}
\item All of the rules   \(\landL\), \(\landR\), \(\lorL\) and \(\lorR\) are height-preserving invertible. \label{inv:bool-rules}
 \item  All of the rules  \(\singDecL\), \(\singDecR\), \(\sepL\),  \(\sepR\), \(\reduceL\), \(\reduceR\), \(\mzSepL\), and \(\mzSepR\)  are height-preserving invertible. \label{inv:imp-rules}
 \item If \(\provable_n \sequent{d,\leftSeq}{\rightSeq}\) for an Euler-Venn diagram \(d\), then
   also \(\provable_n \sequent{\venn{d}, \leftSeq}{\rightSeq}\).\label{inv:ldet}
 \item If \(\provable_n \sequent{d, \leftSeq}{\rightSeq}\) for a pure Euler diagram with one missing zone \(z = (\{n_1, \dots, n_k\}, \{o_1, \dots, o_l\})\), then also \(\provable_n 
        \sequent{\unitaryLiteralPos{o_i}{.15cm},  \leftSeq}{\rightSeq}\) for all \(1 \leq i \leq l\).\label{inv:lidec}
  \end{enumerate}
\end{lemma}
\begin{proof}
  The propositional operator rules are height-preserving invertible as
  shown by Negri et al.~\cite{Negri2001} (Chap. 5, Lemma 5.3.4).
  For the rules \(\singDecL\), \(\singDecR\), \(\sepL\), \(\sepR\), \(\reduceL\),
  \(\reduceR\), \(\mzSepL\) and \(\mzSepR\), similar arguments during
  an induction on the height of the proof yield the result. Case \ref{inv:ldet} and \ref{inv:lidec} can
  be shown by an induction similar to the case of \(\impliesR\). \qed 
\end{proof}

That these rules can be used in an inverse manner is used in the following
lemma, where we connect provability of a sequent \(\sequent{\leftSeq}{\rightSeq}\) within \(\diagName\)
with the provability of the corresponding sequent \(\sequent{\canon{\leftSeq}{}}{\canon{\rightSeq}{}}\)
consisting of the canonical formulas of the antecedent and the succedent. 

\begin{lemma}
  \label{lem:canonical-formula-derive}
  Let \(\sequent{\leftSeq}{\rightSeq}\) be a sequent of compound diagrams. Then \(\sequent{\leftSeq}{\rightSeq}\)
  is provable in \(\diagName\) if, and only if, \(\sequent{\canon{\leftSeq}{}}{\canon{\rightSeq}{}}\) is
  provable in \(\sentName\).
\end{lemma}
\begin{proof}
  Let \(\sequent{\leftSeq}{\rightSeq}\) be provable in \(\diagName\). By Theorem~\ref{thm:soundess},
  the sequent is valid, and hence the sequent \(\sequent{\canon{\leftSeq}{}}{\canon{\rightSeq}{}}\)
  is valid as well. Since \(\sentName\) is complete (cf. Remark~\ref{rem:props-sentential-calculus}),
  the sequent is provable
  in \(\sentName\).

  For the other direction, we proceed by induction on the height \(n\) of the proof of
  \(\sequent{\canon{\leftSeq}{}}{\canon{\rightSeq}{}}\). If \(n = 0\), then
  \(\sequent{\canon{\leftSeq}{}}{\canon{\rightSeq}{}}\) is an axiom
  \(\sequent{p, \leftSeq^\prime}{\rightSeq^\prime, p}\) or an instance of
  \(\falseL\). In the first case, since the only diagram \(D\) with \(\canon{D}{} = p\) is
  a positive literal,  \(\sequent{\leftSeq}{\rightSeq}\) is an axiom as well. Similarly,
  in the second case, it is an instance of \(\falseL\) of \(\diagName\).
  Now assume that the statement is true for all sequents with proofs of height
  less than \(n\). We proceed by a case distinction on the last rule applied
  in the proof of \(\sequent{\canon{\leftSeq}{}}{\canon{\rightSeq}{}}\).

  If the last rule is \(\impliesR\), then the sequent is of the form
  \(\sequent{\canon{\leftSeq}{}}{\canon{\rightSeq^\prime}{}, \canon{D}{}}\),
  where \(D\) is either a compound diagram
  \(D = E \implies F\), a pure Euler diagram \(D = d_e\) with a single missing zone,
  an  Euler-Venn diagram with missing zones and shaded zones \(D = d\), a single negative literal
  for a contour \(c\), or \(D = \true\).  
  In the first case, the premiss is then \(\sequent{\canon{E}{},\canon{\leftSeq}{}}{\canon{F}{}}\),
  which by the induction hypothesis implies that \(\sequent{ E, \leftSeq}{F}\) is
  provable in \(\diagName\). An application of \(\impliesR\) then proves
  \(\sequent{\leftSeq}{\rightSeq}\). Since all cases, where the principal diagram
   is compound are treated exactly like this, we will ignore
  these possibilities in the following. For the case where \(d\) is an Euler-Venn
   diagram, we have \(\canon{d}{} = \euler{d} \implies \venn{d}\).
   and hence the premiss of the last step is
   \(\sequent{\canon{\euler{d}}{}, \canon{\leftSeq}{}}{\canon{\venn{d}}{}}\).
  By the induction hypothesis, we get that
  \(\sequent{\euler{d}, \leftSeq}{\venn{d}}\) is provable, and by applying
  \(\detR\), \(\sequent{\leftSeq}{\rightSeq,d}\) as well. Now assume that
  the principal diagram is a pure Euler diagram \(d_e\) with a single
  missing zone \(z = (\{n_1, \dots, n_k\},\{o_1, \dots, o_l\})\). Hence, the premiss of the last step in \(\sentName\)
  is \(\sequent{\bigland_{1 \leq i \leq k} n_i, \canon{\leftSeq}{}}{\biglor_{1 \leq i \leq l} o_i}\).
  Since both \(\landL\) and \(\lorR\) are height-preserving invertible, the provability of this
  sequent is equivalent to the provability of \(\sequent{n_1, \dots, n_k, \canon{\leftSeq}{}}{o_1, \dots, o_l}\),
  with height less than \(n\).
  Since the canonical formula is only atomic for diagram literals, we have
  that \(\sequent{\unitaryLiteralPos{n_1}{.15cm}, \dots, \unitaryLiteralPos{n_k}{.15cm}, \leftSeq}{
    \unitaryLiteralPos{o_1}{.15cm}, \dots, \unitaryLiteralPos{o_l}{.15cm}}\) is provable by the induction
  hypothesis, and
  hence by applying \(\impDecR\) also \(\sequent{\leftSeq}{\rightSeq, d_e}\). If
  the principal formula was a negative literal for \(c\), then the proven sequent
  is of the form \(\sequent{\canon{\leftSeq}{}}{\canon{\rightSeq}{}, \canon{\unitaryLiteralNeg{c}{.15cm}}{}}\).
  Since \(\canon{\unitaryLiteralNeg{c}{.15cm}}{} = -c = c \implies \false\),
  the premiss is \(\sequent{c, \canon{\leftSeq}{}}{\false}\), which is exactly
  \(\sequent{\canon{\unitaryLiteralPos{c}{.15cm}}{}, \canon{\leftSeq}{}}{\false}{}\). By
  induction hypothesis, we get a proof for \(\sequent{\unitaryLiteralPos{c}{.15cm}, \leftSeq}{}\)
  in \(\diagName\). Thus an application of \(\litR\) yields a proof for
  \(\sequent{\leftSeq}{\rightSeq,\unitaryLiteralNeg{c}{.15cm}}\). Finally,
  if the principal formula  was \(\true\), then \(\canon{D}{} = \true\), and
  an application of \(\trueR\) yields a proof for \(\sequent{\leftSeq}{\rightSeq^\prime,D}\).
  Observe that \(\sequent{\canon{\leftSeq}{}}{\canon{\rightSeq^\prime}{},\canon{\true}{}}\)
  is also provable since the premiss of applying \(\impliesR\) is an instance of
  \(\falseL\).

  If the last application in the proof of \(\sequent{\canon{\leftSeq}{}}{\canon{\rightSeq}{}}\)
  was \(\impliesL\), the arguments are similar, with appropriate applications of
  \(\detL\), \(\impDecL\), \(\litL\), and the invertibility of \(\landR\) and \(\lorL\).

  If the last application was \(\landR\), then the last sequent is of
  the form \(\sequent{\canon{\leftSeq}{}}{\canon{\rightSeq^\prime}{}, \canon{D}{}}\),
  where either \(D = d_e\) is an Euler diagram with more than one missing
  zone, or \(D = d\) is a Venn diagram with exactly one shaded zone. In the first
  case, this means \(\sequent{\canon{\leftSeq}{}}{\canon{\rightSeq^\prime}{}, \bigland_{z^\prime \in \mzones{d_e}}\canon{z^\prime}{m}}\) was proved, and the premisses are
  \(\sequent{\canon{\leftSeq}{}}{\canon{\rightSeq^\prime}{}, \canon{z}{m}}\) and
  \(\sequent{\canon{\leftSeq}{}}{\canon{\rightSeq^\prime}{}, \bigland_{z^\prime \in \mzones{d_e}\setminus \{z\}}\canon{z^\prime}{m}}\) for some \(z \in \mzones{d_e}\). Now consider the
  Euler diagrams \(d_1 = (\contours, \vennzones{\contours} \setminus \{z\})\)
  and \(d_2 = (\contours, (\vennzones{\contours} \setminus \mzones{d}) \cup \{z\})\).
  Then \(\canon{d_1}{} = \canon{z}{m}\) and \(\canon{d_2}{} =  \bigland_{z^\prime \in \mzones{d_e}\setminus \{z\}}\canon{z^\prime}{m}\). Hence, we get by the induction hypothesis that
  \(\sequent{\leftSeq}{\rightSeq^\prime, d_1}\) and \(\sequent{\leftSeq}{\rightSeq^\prime, d_2}\)
  are provable, and thus an application of \(\mzSepR\) yields a proof of
  \(\sequent{\leftSeq}{\rightSeq}\). For the second case, assume \(D=d\) is a
  Venn diagram with exactly one shaded zone \(z = (\{n_1, \dots, n_k\},\{o_1, \dots, o_l\})\), i.e., the
  sequent is in the form \(\sequent{\canon{\leftSeq}{}}{\canon{\rightSeq^\prime}{}, \bigland_{1 \leq i \leq k} n_i \land \bigland_{1 \leq i \leq l} -o_i}\).
  Assume without loss
  of generality that \(n_1\) is part of the outer conjunction, i.e., the conjunction in the succedent is of the form
  \(n_1 \land \left(\bigland_{2 \leq i \leq k} n_i \land \bigland_{1 \leq i \leq l} -o_i\right) \). Hence,
  the premisses are of the form
  \(\sequent{\canon{\leftSeq}{}}{\canon{\rightSeq^\prime}{}, n_1}\) and
  \(\sequent{\canon{\leftSeq}{}}{\canon{\rightSeq^\prime}{}, \bigland_{2 \leq i \leq k} n_i \land \bigland_{1 \leq i \leq l} -o_i}\).
  Since \(\landR\) is height-preserving invertible,
  all sequents of the form \(\sequent{\canon{\leftSeq}{}}{\canon{\rightSeq^\prime}{}, n_i}\)
  and  \(\sequent{\canon{\leftSeq}{}}{\canon{\rightSeq^\prime}{}, -o_i}\)
  are provable with a proof of height less than \(n\).
  From the induction hypothesis, and Remark~\ref{rem:simplified-canonical-literals}, we get that all of the sequents \(\sequent{\leftSeq}{\rightSeq^\prime, \unitaryLiteralPos{n_i}{.15cm}}\) and
  \(\sequent{\leftSeq}{\rightSeq^\prime, \unitaryLiteralNeg{o_i}{.15cm}}\) are provable, and
  hence \(\sequent{\leftSeq}{\rightSeq}\) is provable with an application of \(\singDecR\).

  If the last rule applied in the proof is \(\landL\), the arguments are similar, with
  suited applications of \(\mzSepL\) and \(\singDecL\).

  Now, assume that the last rule applied was \(\lorR\). Then, the only possibility
  is that the principal diagram is a Venn diagram with more than one shaded zone,
  i.e.,   the sequent is \(\sequent{\canon{\leftSeq}{}}{\canon{\rightSeq^\prime}{}, \biglor_{z \in \ezonesOf{d}}\canon{z}{z}}\). So without loss of generality assume that the
  premiss is \(\sequent{\canon{\leftSeq}{}}{\canon{\rightSeq^\prime}{}, \canon{z_i}{z}, \biglor_{z \in \ezonesOf{d}\setminus \{z_i\}}\canon{z}{z}}\). Consider the Venn diagrams
  \(d_1 = (\contours, \zones{d}, \{z_i\})\) and
  \(d_2 = (\contours, \zones{d}, \ezones{d} \setminus \{z_i\})\), and observe
  that \(\canon{d_1}{} = \canon{z_i}{z}\) and
  \(\canon{d_2}{} = \biglor_{z \in \ezonesOf{d}\setminus \{z_i\}}\canon{z}{z}\). That is, by the
  induction hypothesis, we have that \(\sequent{\leftSeq}{\rightSeq^\prime, d_1,d_2}\) is
  provable, and hence by an application of \(\eqDecR\), we can prove
  \(\sequent{\leftSeq}{\rightSeq}\).

  The case for \(\lorL\) is similar, with an appropriate application of
  \(\eqDecL\). \qed
\end{proof}

Since every valid sequent is derivable in \(\sentName\), we get the completeness
result for \(\diagName\) directly from Lemma~\ref{lem:canonical-formula-derive}.

\begin{theorem}[Completeness]
  If \(\sequent{\leftSeq}{\rightSeq}\) is  valid, then \(\sequent{\leftSeq}{\rightSeq}\) is provable.
\end{theorem}

Figure~\ref{fig:ex-deduction} consists of a simple proof containing only Venn diagrams with a single contour.
It shows how disjunction and shaded zones interact. That is, the
presence of several shaded zones can be proven from simpler diagrams.
In particular,
 this proof shows the similarity between the separation rules (\(\eqDecL\) and \(\eqDecR\)) and the rules
 for disjunction. Furthermore, we can see how the rules \(\litL\) and \(\litR\) can be used to
 reduce a sequent with negative literals to an axiom.
 \begin{figure}[h]
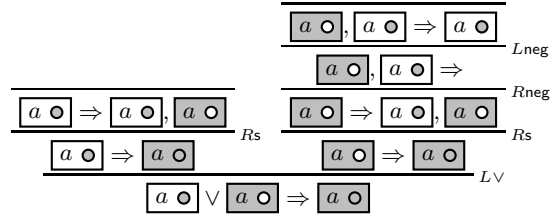

   \centering
   \makebox[\columnwidth]{
     \prftree[r]{\(\scriptstyle \lorL\)}
     {\prftree[r]{\(\scriptstyle \eqDecR\)}
       {\prftree{}{\sequent{\unitaryLiteralPos{a}{.15cm}}{\unitaryLiteralPos{a}{.15cm},\unitaryLiteralNeg{a}{.15cm}}}}
       {\sequent{\unitaryLiteralPos{a}{.15cm}}{\unitaryLEM{a}{.15cm}}}}
     {\prftree[r]{\(\scriptstyle \eqDecR\)}
       {\prftree[r]{\(\scriptstyle \litR\)}
         {\prftree[r]{\(\scriptstyle \litL\)}
           {\prftree{}
             {\sequent{\unitaryLiteralNeg{a}{.15cm},\unitaryLiteralPos{a}{.15cm}}{\unitaryLiteralPos{a}{.15cm}}}}
           {\sequent{\unitaryLiteralNeg{a}{.15cm},\unitaryLiteralPos{a}{.15cm}}{}}}
         {\sequent{\unitaryLiteralNeg{a}{.15cm}}{\unitaryLiteralPos{a}{.15cm},\unitaryLiteralNeg{a}{.15cm}}}}
       { \sequent{\unitaryLiteralNeg{a}{.15cm}}{\unitaryLEM{a}{.15cm} }}
     }
     {\sequent{\unitaryLiteralPos{a}{.15cm} \lor \unitaryLiteralNeg{a}{.15cm}}{\unitaryLEM{a}{.15cm}}}
   }
 \caption{Example of a Simple Proof}
   \label{fig:ex-deduction}
 \end{figure}

 \section{Admissible Rules}
 \label{sec:admissible}
 We show that some rules are admissible. To that end, we define the
 \emph{weight} of diagrams, 
 to order them by
the number of their syntactic elements.

\begin{definition}
  The \emph{weight} \(\weight{d}\) of a diagram is defined inductively. The base cases
  are given by \(    \weight{\false}  = 0\), \(    \weight{\unitaryLiteralPos{c}{.15cm}} = 0\), and
  \(\weight{\unitaryLiteralNeg{c}{.15cm}} = 1\). Otherwise we set
  \begin{align*}
    \weight{d} & =
                 \begin{cases}
                   \size{\ezonesOf{d}} +1 & \text{, if } d \text{ is a Venn diagram}\\
                   \size{\mzones{d}} +1 & \text{, if } d \text{ is a pure Euler diagram}\\
                   \weight{\euler{d}} + \weight{\venn{d}} + 1 & \text{, if } d \text{ is an Euler-Venn diagram}\\
                   \weight{d_1} + \weight{d_2} +1 & \text{, if } d = d_1 \otimes d_2 \text{ for } \otimes \in \{\land, \lor, \implies\}
                 \end{cases}
  \end{align*}
\end{definition}

\begin{lemma}
  \label{lem:general-axiom}
  For any diagram \(D\), the sequent \(\sequent{D, \leftSeq}{\rightSeq,D}\) is provable in \(\diagName\).
\end{lemma}
\begin{proof}
  A straightforward induction on the weight of \(D\). \qed
\end{proof}
\begin{lemma}[Admissibility of Weakening]
  \label{lem:weakening}
  \begin{enumerate*}[label=\roman*)]
  \item  If \(\provable_n \sequent{\leftSeq}{\rightSeq}\), then also \(\provable_n \sequent{D, \leftSeq}{\rightSeq}\). \label{rule:left-weakening}
  \item If \(\provable_n \sequent{\leftSeq}{\rightSeq}\), then also \(\provable_n \sequent{\leftSeq}{\rightSeq, D}\).\label{rule:right-weakening}
  \end{enumerate*}
\end{lemma}
\begin{proof}
  By induction on the height of the proof for \(\sequent{\leftSeq}{\rightSeq}\). For
  \ref{rule:left-weakening}, we can add a new diagram into the antecedent of the
  sequent at the inductive step, since \(\leftSeq\) is kept from the premisses to the
  conclusion. In case~\ref{rule:right-weakening}, this works for most rules as well,
  except, where the succecedent of the premiss is restricted (e.g. \(\litR\)).
  In these cases, the weakening diagram \(D\) is simply added to the multiset \(\rightSeq\) in
  the rule's conclusion. \qed
\end{proof}

\begin{lemma}[Admissibility of Contraction]
  \label{lem:contraction}
  \begin{enumerate*}[label=\roman*)]
\item  If \(\provable_n \sequent{D,D,\leftSeq}{\rightSeq}\), then also \(\provable_n \sequent{D,\leftSeq}{\rightSeq}\).\label{rule:left-contraction}
\item If \(\provable_n \sequent{\leftSeq}{\rightSeq,D,D}\), then also \(\provable_n \sequent{\leftSeq}{\rightSeq,D}\).\label{rule:right-contraction}
\end{enumerate*}
\end{lemma}
\begin{proof}
Both cases can be proven by an induction on the height of
  proofs using Lemma~\ref{lem:inversions} and arguments similar to Negri et al.~\cite{Negri2001}.
  In case~\ref{rule:right-contraction}, the only special case are rules with
  restricted right context in the premisses (e.g. \(\detR\)), where
  the contraction is done by changing the right
  context appropriately. \qed
\end{proof}

\begin{lemma}[Admissibility of Cut]
  If both \({\sequent{\leftSeq}{D,\rightSeq}}\) and \({\sequent{D,\leftSeq^\prime}{\rightSeq^\prime}}\)
  are provable, then also \({\sequent{\leftSeq, \leftSeq^\prime}{\rightSeq,\rightSeq^\prime}}\)
is provable.  
\end{lemma}
\begin{proof}
  We use a semantic proof, employing both soundness and completeness of \(\diagName\). If both
  sequents are provable, they are also valid, by soundness. So choose an arbitrary valuation
  \(\valSing\). Then \(\val{\leftSeq} \leq \val{D} \hor \val{\rightSeq}\) and
  \(\val{D} \hand \val{\leftSeq^\prime} \leq \val{\rightSeq^\prime}\). Now we have
  \(\val{\leftSeq} \hand \val{\leftSeq^\prime} \leq  (\val{D} \hor \val{\rightSeq}) \hand \val{\leftSeq^\prime}
  = (\val{D} \hand \val{\leftSeq^\prime}) \hor (\val{\rightSeq} \hand \val{\leftSeq^\prime})
  \leq \val{\rightSeq^\prime} \hor (\val{\rightSeq} \hand \val{\leftSeq^\prime})
  \leq \val{\rightSeq^\prime} \hor \val{\rightSeq}
  \). These relations are due to the first premiss, distributivity, the second premiss and
  the fact \(a \hand b \leq a\), respectively.  
  Since \(\valSing\) was arbitrary, \(\sequent{\leftSeq, \leftSeq^\prime}{\rightSeq,\rightSeq^\prime}\)
  is valid, and due to the completeness of \(\diagName\), we have that \(\sequent{\leftSeq, \leftSeq^\prime}{\rightSeq,\rightSeq^\prime}\)
  is provable. \qed
\end{proof}

\begin{remark}
  It is also possible to prove cut admissibility with a purely syntactic
  argument by adapting the inductive proof for the system \(\sentName\) given
  by Negri et al.\cite{Negri2001}. The proof consists of a replacement
  of each cut application with a derivation, where each cut either
  posesses a lower cut-height, or the weight of the cut diagram is lower.
  Within that proof, most cases
  are straightforward, where \(\singDecL\), \(\singDecR\), \(\reduceL\),
  \(\reduceR\), \(\mzSepL\) and \(\mzSepR\) are treated similarly
  to the rules \(\landL\) and \(\landR\), while \(\eqDecL\) and \(\eqDecR\)
  play roles similar to \(\lorL\) and \(\lorR\). The rules
  \(\litL\), \(\litR\), \(\detL\), \(\detR\), \(\impDecL\) and \(\impDecR\)
  need special attention, since they restrict the succedent in the premiss.
  However, the proof proceeds in these cases along the lines of the
  the treatment of \(\impliesL\) and
  \(\impliesR\) in \(\sentName\). While the number of cases
  to consider  increases, the arguments and constructions are similar.
  As an example, we present the case where the cut formula is principal
  in both premisses, and is a negative literal. That is, we have a derivation
  of the following form:
  \begin{align*}    
  \prftree[l]{\(\scriptstyle \cut\)}
  {\prftree[l]{\(\scriptstyle \litR\)}
    {\sequent{\unitaryLiteralPos{c}{.15cm},\leftSeq}{}}
    {\sequent{\leftSeq}{\rightSeq, \unitaryLiteralNeg{c}{.15cm}}}
  }
  {\prftree[r]{\(\scriptstyle \litL\)}
    {\sequent{ \unitaryLiteralNeg{c}{.15cm},\leftSeq^\prime}{\unitaryLiteralPos{c}{.15cm}}}
    {\sequent{ \unitaryLiteralNeg{c}{.15cm},\leftSeq^\prime}{\rightSeq^\prime}}
  }
  {\sequent{\leftSeq, \leftSeq^\prime}{\rightSeq, \rightSeq^\prime}}
  \end{align*}

  Observe that the cut-height of this cut application is \(m+n+2\), where
  \(m\) is the height of the proof of the left premiss and \(n\)
  the height of the proof of the right premiss. Then, we can replace this
  derivation with the following.
  \begin{align*}    
    \prftree[l]{\(\scriptstyle \weakR, \contractL\)}
    {\prftree[l]{\(\scriptstyle \cut\)}
    {\prftree[l]{\(\scriptstyle \cut\)}
    {\prftree[l]{\(\scriptstyle \litR\)}
    {\sequent{\unitaryLiteralPos{c}{.15cm},\leftSeq}{}}
    {\sequent{\leftSeq}{\rightSeq, \unitaryLiteralNeg{c}{.15cm}}}
    }
    {\sequent{ \unitaryLiteralNeg{c}{.15cm},\leftSeq^\prime}{\unitaryLiteralPos{c}{.15cm}}}
    {\sequent{\leftSeq, \leftSeq^\prime}{\rightSeq,\unitaryLiteralPos{c}{.15cm} }}
    }
    {\sequent{\unitaryLiteralPos{c}{.15cm},\leftSeq}{}}
    {\sequent{\leftSeq,\leftSeq^\prime,\leftSeq}{\rightSeq}}
    }
    {\sequent{\leftSeq,\leftSeq^\prime}{\rightSeq,\rightSeq^\prime}}
  \end{align*}

  In this derivation, the uppermost cut has a lower cut-height, while
  the second cut uses a cut diagram of lower weight. Here, it is crucial that
  the negative literal has a higher weight than the positive literal. The last step
  in the derivation is a sequence of weakening and contraction. The
  treatment of the other cases is analogous.
\end{remark}


 \begin{figure}
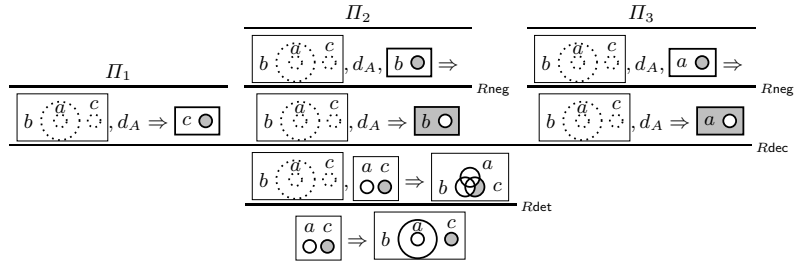

   \begin{center}
     \scalebox{.85}{
     \prftree[r]{\(\scriptstyle \detR\)}
     {\prftree[r]{\(\scriptstyle \singDecR\)}
       {\prftree[r]{}
         {\Pi_1}
         {\sequent{\subsdisjEDCstr{a}{b}{c}{.2cm}, d_A}{\unitaryLiteralPos{c}{.2cm}}}
       }
       {\prftree[r]{\(\scriptstyle \litR\)}
         {\prftree[r]{}
           {\Pi_2}
           {\sequent{\subsdisjEDCstr{a}{b}{c}{.2cm}, d_A, \unitaryLiteralPos{b}{.2cm}}{}}
           }
         {\sequent{\subsdisjEDCstr{a}{b}{c}{.2cm}, d_A}{\unitaryLiteralNeg{b}{.2cm}}}
       }
       {\prftree[r]{\(\scriptstyle \litR\)}
         {\prftree[r]{}
           {\Pi_3}
           {\sequent{\subsdisjEDCstr{a}{b}{c}{.2cm}, d_A, \unitaryLiteralPos{a}{.2cm}}{}}
         }
         {\sequent{\subsdisjEDCstr{a}{b}{c}{.2cm}, d_A}{\unitaryLiteralNeg{a}{.2cm}}}
       }
       {\sequent{\subsdisjEDCstr{a}{b}{c}{.2cm}, \disjointEDEx{a}{c}{.2cm}}{\exVenn}}
     }
     {\sequent{\disjointEDEx{a}{c}{.2cm}}{\subsdisjEDEx{a}{b}{c}{.2cm}}}
     }
   \end{center}   
   \caption{Proof using Euler-Venn diagrams}
   \label{fig:example_proof2}
 \end{figure}
 
A  derivation that uses all three types of diagrams can be found in Fig.~\ref{fig:example_proof2}. We 
explain parts
 of the proof from bottom to top. The last applied rule detaches the pure Euler part from
the Venn part of the succedent,
   so that we can then decompose the single shaded zone into literals.
 This splits the proof into three branches, which we treat in the sub-derivations \(\Pi_1\),
 \(\Pi_2\) and \(\Pi_3\), respectively. For reasons of brevity, we use the abbrevations for
diagrams as shown in Table~\ref{tab:example_abbrev}. Now, the two right proof branches
 contain a negative literal in the succedent, which we move to the antecedent with an
 application of \(\litR\).
 Then, all three proof branches proceed similarly:
 we reduce the pure Euler diagram \(d^\ast_C\) into smaller diagrams. 
 The
 set of missing zones is \(\mzones{d^\ast_C} = \{(\{a\},\{b,c\}), (\{a,c\}, \{b\}), (\{b,c\}, \{a\}), (\{a,b,c\}, \emptyset) \}\),
 and  each of these missing zones has at least one adjacent missing zone. For example,  \(\adj{(\{a\},\{b,c\})}{c} = (\{a,c\},\{b\})\).
 In particular, the reduction of \(d^\ast_C\) with respect
 to any of the contours
 \setlength\intextsep{0pt}
\begin{wraptable}[6]{l}{.4\textwidth}
   \caption{Diagram Abbreviations}
   \label{tab:example_abbrev}
   \begin{tabular}{ccc}
     \subsdisjEDEx{a}{b}{c}{.25cm} &
     \subsdisjEDCstr{a}{b}{c}{.25cm} & 
     \disjointEDEx{a}{c}{.25cm} \\[.75em]                                     
\(d_C\)  & \(d_C^\ast\)  & \(d_A\) \end{tabular}
\end{wraptable}
\setlength{\intextsep}{\savedintextsep}
 \(a\), \(b\) and \(c\) still contains missing zones. It is easy to check that the
 three diagrams shown in the derivations are indeed these reductions. Then, \(\Pi_1\) proceeds by detaching the
 Euler and Venn aspects of the diagram \(d_A\), which immediately closes the left branch, due to Lemma~\ref{lem:general-axiom}.
 The right branch ends in an axiom after decomposing the single shaded zone in the antecedent. Within
 \(\Pi_2\) there is a similar structure, denoted by the derivation \(\Pi_1^\prime\), where
 the antecedent contains slightly different diagrams, but the application of rules is similar. The other
 branches proceed similarly. This example shows, how the reduction rules lead to smaller diagrams, and,
 as we claim, better readable diagrams, due to the reduced clutter~\cite{John2005}. Furthermore, it shows
 how the admissible rules may reduce the size of the proofs, here in the form of the generalised axioms
 proven admissible in Lemma~\ref{lem:general-axiom}.

  \begin{figure}
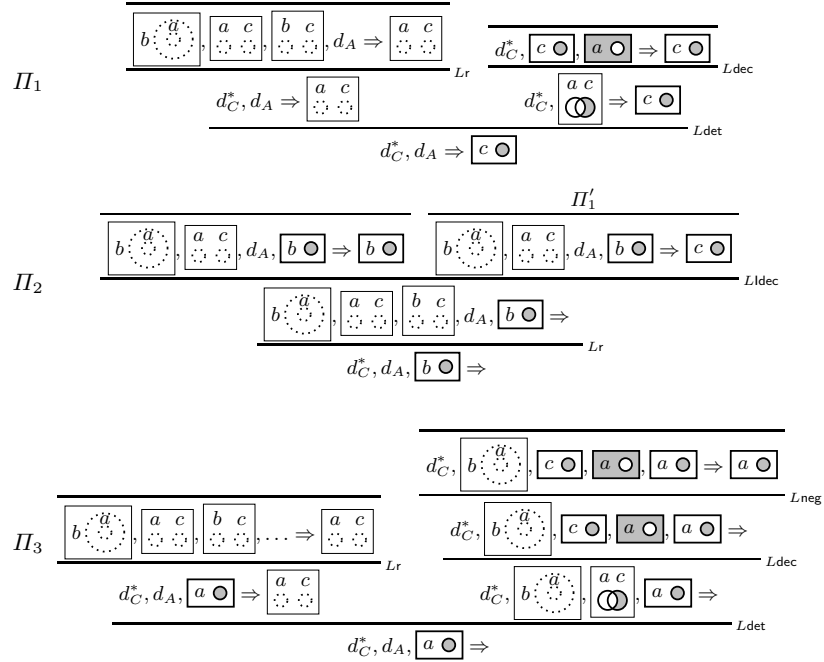

   \begin{tabular}{cc}
     \(\Pi_1\)    &
     \scalebox{.85}{
                    \parbox[c]{\hsize}{
                    \prftree[r]{\(\scriptstyle \detL\)}
                    {\prftree[r]{\(\scriptstyle \reduceL\)}
                    {\prftree{}{\sequent{\subsetEDCstr{b}{a}{.2cm}, \disjointEDCstr{a}{c}{.2cm}, \disjointEDCstr{b}{c}{.2cm},d_A}{\disjointEDCstr{a}{c}{.2cm}}}}
                    {\sequent{d_C^\ast, d_A}{ \disjointEDCstr{a}{c}{.2cm}}}
                    }
                    {
                    \prftree[r]{\(\scriptstyle \singDecL\)}
                    {\prftree{}
                    {\sequent{d_C^\ast, \unitaryLiteralPos{c}{.2cm}, \unitaryLiteralNeg{a}{.2cm}}{\unitaryLiteralPos{c}{.2cm}}}
                    }
                    {\sequent{d_C^\ast, \cNotAVDEx}{\unitaryLiteralPos{c}{.2cm}}}
                    }                    
                    {\sequent{d_C^\ast, d_A}{\unitaryLiteralPos{c}{.2cm}}{}}
                    }
                    }\vspace{1em}
     \\
     \(\Pi_2\)   &     \scalebox{.85}{

                   \parbox[c]{\hsize}{
                   \prftree[r]{\(\scriptstyle \reduceL \)}
                   {\prftree[r]{\(\scriptstyle \impDecL\)}
                   {\prftree{}{\sequent{\subsetEDCstr{b}{a}{.2cm}, \disjointEDCstr{a}{c}{.2cm},d_A, \unitaryLiteralPos{b}{.2cm} }{\unitaryLiteralPos{b}{.2cm}}}}
                   {\prftree[r]{}
                   {
                   \Pi_1^\prime
                   }
                   {\sequent{\subsetEDCstr{b}{a}{.2cm}, \disjointEDCstr{a}{c}{.2cm},d_A, \unitaryLiteralPos{b}{.2cm} }{\unitaryLiteralPos{c}{.2cm}}}
                   }
                   {\sequent{\subsetEDCstr{b}{a}{.2cm}, \disjointEDCstr{a}{c}{.2cm}, \disjointEDCstr{b}{c}{.2cm},d_A, \unitaryLiteralPos{b}{.2cm} }{}}}
                   {\sequent{d_C^\ast, d_A, \unitaryLiteralPos{b}{.2cm}}{}}
                   }
                   }\vspace{2em}
                   
     \\
     \(\Pi_3\) &     \scalebox{.85}{

                 \parbox[c]{\hsize}{
                 \prftree[r]{\(\scriptstyle \detL\)}                 
                 {\prftree[r]{\(\scriptstyle \reduceL\)}
                 {\prftree{}
                 {\sequent{\subsetEDCstr{b}{a}{.2cm}, \disjointEDCstr{a}{c}{.2cm}, \disjointEDCstr{b}{c}{.2cm},\dots}{\disjointEDCstr{a}{c}{.2cm}}}}
                 {\sequent{d_C^\ast,d_A, \unitaryLiteralPos{a}{.2cm}}{\disjointEDCstr{a}{c}{.2cm}}}
                 }
                 {\prftree[r]{\(\scriptstyle \singDecL\)}
                 {\prftree[r]{\(\scriptstyle \litL\)}
                 {\prftree{}
                 {\sequent{d_C^\ast,\subsetEDCstr{b}{a}{.2cm},\unitaryLiteralPos{c}{.2cm},\unitaryLiteralNeg{a}{.2cm},\unitaryLiteralPos{a}{.2cm}}{\unitaryLiteralPos{a}{.2cm}}}
                 }
                 {\sequent{d_C^\ast,\subsetEDCstr{b}{a}{.2cm},\unitaryLiteralPos{c}{.2cm},\unitaryLiteralNeg{a}{.2cm},\unitaryLiteralPos{a}{.2cm}}{}}
                 }
                 {\sequent{d_C^\ast,\subsetEDCstr{b}{a}{.2cm}, \cNotAVDEx,\unitaryLiteralPos{a}{.2cm}}{}}                                 
                 }
                 {\sequent{d_C^\ast, d_A,\unitaryLiteralPos{a}{.2cm}}{}}                 
                 }
                 }
   \end{tabular}
   \caption{Auxiliary Derivations for Fig.~\ref{fig:example_proof2}}
   \label{fig:aux_deriv}
 \end{figure}



\section{Conclusion}
\label{sec:conc}
In this paper, we presented an intuitionistic interpretation
of Euler-Venn diagrams, based on a semantics of Heyting algebras.
We then defined a cut-free sequent calculus \(\diagName\), which we have proven
to be sound and complete with respect to this semantics. Furthermore,
we have shown that the structural rules of contraction, weakening
and cut are admissible.

For this visualisation, we deviated from classical
Euler-Venn diagrams in two ways: we did not treat missing zones and shaded
zones as equivalent, and we  introduced the new syntactic element of
dashed contours.

The first deviation is due to the basic
restrictions of intuitionistic reasoning. More specifically,
intuitionistic implication cannot be treated as an abbreviation of the other
operators. To have a syntax explicitly for implications,  we need to
increase the number of distinct syntactic elements
of Euler-Venn diagrams. Hence, distinguishing these two elements
is a natural choice. Of course, it can be argued that the choice
we made is not the correct one, and that shading should be
used to reflect implications. However, we think that
since the representation of missing zones (or rather their absence)
introduces a direction into the diagram, in the form of
inclusions, this choice is justified.

The introduction of dashed diagrams is more debatable. Arguably, the need
for distinguishing pure Euler diagrams by dashing arises, since we interpret
the missing zones of Euler-Venn diagrams as a kind of ``constructive precondition'' for
the construction of the elements denoted by the shaded zones. That is, in the constructive
interpretation of intuitionistic reasoning, an Euler-Venn diagram
means that, given a construction as indicated by the missing zones, we have another
construction for
the assertions given by the shaded zones. Hence, there is an additional
implication within the semantics of Euler-Venn diagrams, as can also be seen in 
the rules of \(\diagName\) to detach the pure Euler aspects from the Venn aspects
of a diagram. These rules behave similarly to the rules for implication in 
sentential intuitionistic sequent calculus.

However, the introduction of new syntactic elements is necessary, due
to the independence of the operators, and the restrictive nature
of Euler-Venn diagrams makes this need even more overt. Compare for example
the intuitionistic systems based on Existential Graphs (EGs).
While the operations in classical EGs are denoted by juxtaposition and  cuts,
reflecting conjunction and negation, respectively, the assertive graphs \cite{Bellucci2018}
 explicitly introduce notation for disjunction, and also treat the
 ``scroll'' as a distinct element. Similarly, the intuitionistic
 EGs \cite{Ma2019} include the notion of \emph{\(n\)-scrolls} for each
 \(n > 0\).

We think that our system stretches the idea of Euler-Venn diagrams quite
far. In particular, logics that need even more independent operators,
for example substructural logics and modal logics, may not be well-matched
for such a diagrammatic system. While it may be possible to define
such an interpretation, the type of new syntactic elements is far from
obvious, if we want to keep the diagrammatic structure of
 Euler-Venn diagrams. Of course, it is always possible to add new
operators to the compound part of the reasoning system, but
we think that such an addition misses the point of a
diagrammatic reasoning system. 

Still, there are future directions this work can be taken into.
For example, our sequent calculus resembles sentential sequent
calculus, while typical Euler-Venn reasoning
systems work by adding syntax to single
diagrams, and then removing unnecessary parts \cite{Burton2012}. It is interesting to see, if
we can define such a system for intuitionistic
Euler-Venn diagrams. We assume that for the rules
to introduce and remove contours, or to copy contours
from one diagram into another, the reduction of
a pure Euler diagram (cf. Def~\ref{def:reduction} and
Lemma~\ref{lem:sem-reduction}) will play a significant
role.


\bibliographystyle{splncs03}
\bibliography{lit}

\end{document}